\documentclass[10pt]{article}
\usepackage{caption}

\usepackage{subcaption}
\usepackage[utf8]{inputenc}

\usepackage{amssymb}
\usepackage{amsfonts}
\usepackage{amsmath}
\usepackage{amsthm}
\usepackage{mathtools}
\usepackage{bbm}
\usepackage{graphicx}
\usepackage{indentfirst}
\usepackage{enumerate}
\usepackage{url}

\usepackage{lmodern}
\usepackage{newpxtext}

\usepackage[symbol]{footmisc}

\usepackage[normalem]{ulem}

\usepackage{tikz}
\usetikzlibrary{quantikz2}
\definecolor{darkblue}{rgb}{0.15,0.35,0.55}
\definecolor{reddish}{rgb}{.8, 0.2, 0.2}

\usepackage{pgfplots}
\definecolor{plotblue}{RGB}{0,120,200}
\definecolor{plotgreen}{RGB}{0,155,130}
\definecolor{plotorange}{RGB}{240,120,50}
\definecolor{plotmagenta}{RGB}{240,50,120}
\definecolor{plotgray}{RGB}{128,128,128}
\definecolor{plotcyan}{RGB}{50,190,240}
\definecolor{plotred}{RGB}{205,50,15}
\usepgfplotslibrary{fillbetween}

\usepackage[backend=bibtex, style = alphabetic, doi = false, url = true, isbn =
false, maxbibnames = 6]{biblatex}
\bibliography{paperRef}
\renewbibmacro{in:}{}
\newbibmacro{string+doi}[1]{\iffieldundef{doi}{#1}{\href{https://dx.doi.org/\thefield{doi}}{#1}}}
\DeclareFieldFormat{title}{\usebibmacro{string+doi}{\mkbibemph{#1}}}
\DeclareFieldFormat[article]{title}{\usebibmacro{string+doi}{\mkbibquote{#1}}}
\DeclareFieldFormat[incollection]{title}{\usebibmacro{string+doi}{\mkbibquote{#1}}}
\DeclareFieldFormat[inproceedings]{title}{\usebibmacro{string+doi}{\mkbibquote{#1}}}

\usepackage[pdfpagelabels, linktocpage=true]{hyperref}
\hypersetup{
colorlinks=true,
citecolor=darkblue,
linkcolor=reddish,
urlcolor=darkblue,
pdfauthor={},
pdftitle={},
pdfsubject={}
}

\long\def\ca#1\cb{} 

\newcommand{\becs}{\begin{cases}}
\newcommand{\bem}{\begin{matrix}}

\newcommand{\dya}[1]{|#1\rangle\langle#1|}
\newcommand{\dyad}[2]{|#1\rangle\langle#2|}

\newcommand{\encs}{\end{cases}}
\newcommand{\enm}{\end{matrix}}

\newcommand{\mted}[3]{\langle#1|#2|#3\rangle }

\newcommand{\ot}{\otimes }

\newcommand{\Tr}{{\rm Tr}}

\newcommand{\AC}{{\mathcal A}}

\newcommand{\DC}{{\mathcal D}}
\newcommand{\EC}{{\mathcal E}}
\newcommand{\FC}{{\mathcal F}}
\newcommand{\GC}{{\mathcal G}}
\newcommand{\HC}{{\mathcal H}}
\newcommand{\IC}{{\mathcal I}}

\newcommand{\MC}{{\mathcal M}}
\newcommand{\NC}{{\mathcal N}}

\newcommand{\rB}{\textbf{r}}

\newcommand{\Ibb}{\mathbb{I}}

\newcommand{\al}{\alpha }
\newcommand{\bt}{\beta }

\newcommand{\dl}{\delta }

\newcommand{\ep}{\epsilon}

\newcommand{\lm}{\lambda }
\newcommand{\Lm}{\Lambda }

\newcommand{\sg}{\sigma }

\newcommand{\ketbra}[1]{\ket{#1}\bra{#1}}

\theoremstyle{remark}

\newtheorem{lemma}{Lemma}

\usepackage[linesnumbered]{algorithm2e}

\usetikzlibrary{arrows}
\usetikzlibrary{decorations.pathreplacing,calligraphy}
\usepackage{verbatim}

\usepackage{bbm}
 
\usepackage[affil-it]{authblk}

\title{Entanglement sharing across a damping-dephasing channel~\footnote{This
paper was presented in part at 2024 IEEE International
Symposium on Information Theory (ISIT)~\cite{10619242}}}
\author[1,2]{Vikesh Siddhu}
\author[3]{Dina Abdelhadi}
\author[2,4]{Tomas Jochym-O’Connor}
\author[2]{John Smolin}
\affil[1]{IBM Research, IBM Research India, India}
\affil[2]{IBM Quantum, IBM T.J. Watson Research Center, Yorktown Heights, NY, USA}
\affil[3]{School of Computer and Communication Sciences, EPFL, Switzerland}
\affil[4]{IBM Quantum, Almaden Research Center, San Jose, CA, USA}

\date{28 Jul 2025}
\begin{document}

\maketitle

\begin{abstract}
    Entanglement distillation is a fundamental information processing task
    whose implementation is key to quantum communication and modular quantum
    computing.  Noise experienced by such communication and computing platforms
    occurs not only in the form of Pauli noise such as dephasing (sometimes
    called $T_2$) but also non-Pauli noise such as amplitude damping (sometimes
    called $T_1$). We initiate a study of practical and asymptotic distillation
    over the so-called joint damping-dephasing noise
    channel. In the practical setting, we propose a distillation scheme that
    completely isolates away the damping noise. In the asymptotic setting we
    derive lower bounds on the entanglement sharing capacities including the
    coherent and reverse coherent information.  Like the protocol achieving the
    reverse coherent information, our scheme uses only backward classical
    communication. However, for realistic damping noise ($T_1 \neq 2T_2$) our
    strategy can exceed the reverse coherent strategy.  In the forward communication setting we numerically
    exceed the single-letter coherent information strategy by observing the
    channel displays non-additivity at the two-letter level. The work shows
    non-additivity can also be found in realistic noise models with magnitudes
    of non-additivity similar to those found in more idealized noise channels.
\end{abstract}

\section{Introduction}

Entanglement is a fundamental quantum resource which enables quantum
computation and communication~\cite{HorodeckiHorodeckiEA09a}. A key goal of
quantum information theory~\cite{BennettShor98} is to find, understand, and
achieve maximum rates at which this resource can be shared noiselessly across
asymptotically many independent uses of a noisy quantum channel. This maximum
rate, called a channel's capacity, can be defined depending on the availability
of classical communication between the sender and receiver of the quantum
channel~\cite{BennettShor04, Wilde17a}.
Availability of free noiseless forward~(from sender to receiver),
backward~(only from receiver to sender), and two-way classical communication to
assist transmission across the quantum channel define the channel's forward,
backward, and two-way entanglement sharing~(also called entanglement
distillation) capacities, $\EC_{\rightarrow}$, $\EC_{\leftarrow}$, and
$\EC_{\leftrightarrow}$, respectively~\cite{Leung08}.
Capacity in the absence of forward communication, $\EC$, is known to equal
$\EC_{\rightarrow}$ and it follows that $\EC \leq \EC_{\leftarrow} \leq
\EC_{\leftrightarrow}$~\cite{BarnumKnillEA00, GarciaPatronPirandolaEA09}. However, a tight understanding of each of these
capacities has thus far eluded the research community. For instance, obtaining
easy to compute entropic expressions for each capacity has been hard and
protocols for achieving these capacities are seldom known for general channels.

A channel $\NC$'s coherent information $I_c(\NC)$ is an entropic
expression~(for def. see eq.~\eqref{eq:ciDef}) that bounds its capacity $\EC(\NC)$ from
below, $I_c(\NC) \leq \EC(\NC)$~\cite{SchumacherNielsen96}.  Here, equality
holds for the class of (anti)~degradable~\cite{DevetakShor05} channels but in
general this inequality can be strict~\cite{DiVincenzoShorEA98,
CubittElkoussEA15} and $\EC(\NC)$ is given by $I_c(\NC^{\ot k }) / k$ in the
limit $k$ tends to infinity, where $\NC^{\ot k}$ represent $k$ joint uses of
$\NC$~\cite{Lloyd97, Shor02a, Devetak05}.
Capacity $\EC$ also equals the maximum achievable rate for quantum error
correction, thus a channel's coherent information also provides a useful bound
on the ability of error correcting codes to remove noise introduced by a
channel.
For example, the well-known hashing rate~\cite{BennettDiVincenzoEA96} used to
benchmark quantum error correcting codes under depolarizing noise
$\Lm$~\cite{GhoshFowlerEA12, BonillaAtaidesTuckettEA21} is nothing but the
channel's coherent information $I_c(\Lm)$.
For most other types of channels~(with the notable exception of degradable
channels), the channel coherent $I_c$ is not easy to determine. 
This difficulty not only precludes one from finding an important lower bound on
the a channel's capacity but also hinders finding an important
benchmark for quantum error correction across the channel.

Error correcting codes are not only useful for protecting quantum information
but they are also valuable for distributing entanglement. Rates for
distribution can be improved beyond $\EC$ by simply adding backward classical
communication~\cite{LeungLimShor09}, however even this simple addition
complicates the discussion of entanglement sharing capacities.  Analogously to
the capacity~$\EC$, we can bound the backward capacity~$\EC_{\leftarrow}$ by an
entropic quantity called the reverse coherent information~$I_r$ (for def. see
eq.~\eqref{eq:rciDef})~\cite{HorodeckiHorodeckiEA00a, DevetakWinter05,
DevetakJungeEA06, GarciaPatronPirandolaEA09}. Since $I_r$ has several
remarkable properties such as concavity and additivity, it can be used more
easily even though it doesn't satisfy data-processing.  Finally, while not the
primary focus in this work, the two-way capacity $\EC_{\leftrightarrow}$, is
even less understood than its one-way counterpart.

One way to make progress in understanding one and two-way capacities is to
study specific channels. Indeed qubit Pauli channels have been extensively
studied and a variety of interesting insights, such as super-additivity of
$I_c$ in the one-way setting~\cite{SmithSmolin07, FernWhaley08,
BauschLeditzky21}, and a variety of innovative
protocols~\cite{VollbrechtVerstraete05, HostensDehaeneEA06} in the two-way
setting have been found this way. Unfortunately, progress on two-way protocols
from studying Pauli channels has been slow and realistic noise models are
typically non-Pauli and non-unital.  There is a need to use such realistic
noise models for understanding error correction across these models and also
for developing new ideas for backward and two-way quantum capacities.  Without
getting a better understanding of the one and two-way capacities of quantum
channels we not only fail to build realistic expectations for modular quantum
computing and long-range quantum communication but also remain partial in our
understanding of quantum entanglement itself.

In this work we propose a way forward to study rates for sharing entanglement
across a physically relevant channel~\cite{SarvepalliKlappeneckerEA08,
AliferisBritoEA09, GhoshFowlerEA12}~(considered in the context of
fault-tolerant quantum computation). This channel is obtained by concatenating
two well-studied channels. 
The first channel is a qubit dephasing and the second is a qubit amplitude
damping~\cite{GiovannettiFazio05}.
For both channels the reverse coherent information is known to achieve the
largest known rate for sharing entanglement using backward~(or even two-way)
classical communication~\cite{GarciaPatronPirandolaEA09}.
Nonetheless, we find that the joint channel benefits from our proposed backward
only protocol that exceeds the channel's reverse coherent information.  One key
aspect of our strategy is to completely isolate the amplitude damping component
of damping-dephasing noise. In this way, the work provides a single-shot
protocol to remove damping noise.
The work paves the way for benchmarking error correction across this
damping-dephasing channel by providing the analog of the hashing rate. It goes
a step further to improve upon this rate by observing the channel's coherent
information displays non-additivity. Even though the noise is more realistic,
the non-additivity found is simple~(occurring at the two-letter level) and with
magnitude comparable to those found in more idealized noise models. 

\section{Preliminaries}
\subsection{Kraus representation}
Let $\HC_a$, $\HC_b$, and $\HC_c$ be finite dimensional Hilbert spaces with
standard orthonormal basis $\{\ket{k}_a\}, \{\ket{j}_b\},$ and $\{\ket{i}_c\}$
respectively.
An isometry $E : \HC_a \mapsto \HC_b \ot \HC_c$~(here $\ot$ represents tensor
product) satisfies $E^{\dag}E = \Ibb_a$, the identity on $\HC_a$, where $\dag$
represents conjugate transpose. This isometry can be expanded as $E = \sum_i K_i \ot
\ket{i}_c = \sum_j L_j \ot \ket{j}_b$~(notice $\mted{j}{K_i}{k} =
\mted{i}{L_j}{k}$) and generates a pair of channels
\begin{align}
    \begin{aligned}
    \NC(M) &= \Tr_{c}(EME^{\dag}) = \sum_i K_i M K_i^{\dag},
        \\
    \NC^c(M) &= \Tr_{b}(EME^{\dag}) = \sum_j L_j M L_j^{\dag}
    \label{eq:chanDef}
    \end{aligned}
\end{align}
from $a$ to $b$ and $a$ to $c$, respectively, with Kraus operators $\{K_i\}$
and $\{L_j\}$, respectively, here $M$ is any linear operator on $\HC_a$. One
says $\NC^c$ is the complement of $\NC$ and vice-versa.  If $\NC^c$ can be
obtained from the output of $\NC$ by applying some channel $\MC$, i.e., $\MC
\circ \NC = \NC^c$, then $\NC$ is said to be degradable and $\NC^c$
anti-degradable.  An {\em entanglement-breaking}~(EB)
channel~\cite{HorodeckiShorEA03a} $\NC$ takes the form $\NC(M) = \sum_i \Tr(M
V_i) \rho_i$ where $\{V_i\}$ are a collection of positive-semidefinite~(PSD)
operators that sum to the identity, and $\rho_i$ are density operators~(PSD and
unit trace $\Tr(\rho_i) = 1$).

\subsection{Entropic quantities: coherent information and capacities}
The von-Neumann entropy of a density operator $\rho$,
\begin{equation}
    S(\rho) := - \Tr(\rho \log \rho) = - \sum_i \lm_i \log \lm_i,
    \label{eq:entroVon}
\end{equation}
where $\{\lm_i\}$ are eigenvalues of $\rho$ and we use $\log$ base 2 by
default.
The coherent information of a channel $\NC$ at an input density operator
$\rho$, 
\begin{equation}
    I_c(\NC, \rho) = S \big( \NC (\rho) \big) - S \big( \NC^c(\rho)\big),
    \label{eq:ciDef}
\end{equation}
maximized over density operators $\rho$ gives the channel coherent information
$I_c(\NC)$.  This maximization is generally non-convex, except with $\NC$ is
degradable and $I_c(\NC, \rho)$ becomes concave in
$\rho$~\cite{YardHaydenEA08}.  A channel $\NC$'s entanglement sharing capacity
is given by the limit~\cite{Lloyd97, Shor02a, Devetak05},
\begin{equation}
    \EC(\NC) = \lim_{k \mapsto \infty} \frac{1}{k} I_c(\NC^{\ot k}, \rho), 
    \quad k \in \mathbf{N},
    \label{eq:eCap}
\end{equation}
which equals $I_c(\NC)$ when $\NC$ is either degradable or anti-degradable.
For anti-degradable channels $\EC(\NC) = 0$. 
The reverse coherent information of a channel $\NC$ at an input density
operator $\rho$, 
\begin{equation}
    I_r(\NC, \rho) = S (\rho) - S \big( \NC^c(\rho)\big),
    \label{eq:rciDef}
\end{equation}
represents an achievable rate for sharing entanglement across $\NC$ using
only backward classicaly communication from $b$ to $a$ in addition to any pre-agreed
classical communication between $b$ and $a$ prior to using the channel.
Since $I_r(\NC, \rho)$ is concave in $\rho$ for all $\NC$, it can be maximized
over density operators $\rho$ with relative ease to obtain the reverse coherent
information $I_r(\NC)$.  This reverse coherent information is additive
$\IC_r(\NC^{\ot k}) = k \IC_r(\NC) \ \forall \ k \in
\mathbb{N}$~\cite{GarciaPatronPirandolaEA09}.

\subsection{$\log$-singularity Method(s)}
\label{App.lSing}
Here we review a $\log$-singularity method~\cite{Siddhu21,SinghDatta22}) and
state a minor extension of the method to analyse the reverse coherent
information.

Let $0 \leq \ep \leq 1$ be a real
parameter and $\rho(\ep)$ be a one-parameter family of density operators with
von-Neumann entropy $S(\ep): = S\big( \rho(\ep) \big)$.  If one or several
eigenvalues of $\rho(\ep)$ increase linearly from zero to leading order in
$\ep$ then a small change in $\ep$ from zero changes $S(\ep)$ by $x |\ep \log
\ep|$, $x>0$, and $S(\ep)$ is said to have an $\ep$ $\log$-singularity with
rate $x$.

Let $\NC$ be a channel with complement $\NC^c$. These channels $\NC$ and $\NC^c$
map an input density operator $\rho_a(\ep)$ to $\rho_b(\ep) := \NC \big(
\rho_a(\ep) \big)$ and $\rho_c(\ep) = \NC^c \big( \rho_a(\ep) \big)$,
respectively. Let $S_a(\ep), S_b(\ep),$ and $S_c(\ep)$ denote the von-Neumann
entropies of $\rho_a(\ep), \rho_b(\ep)$, and $\rho_c(\ep)$, respectively.
At $\rho_a(\ep)$, let $I_c(\ep):= S_b(\ep) - S_c(\ep)$ denote the coherent information 
and $I_r(\ep):= S_a(\ep) - S_c(\ep)$ denote the reverse coherent information.

If there is an input $\rho_a(\ep)$ such that at $\ep = 0$, $I_c(\ep) = 0$ and
$S_b(\ep)$ has an $\ep$ $\log$-singularity with a higher rate than $S_c(\ep)$,
i.e., $S_b(\ep)$ has a stronger $\log$-singularity than $S_c(\ep)$, then
$I_c(\NC) > 0$. This statement showing positivity of a channel's coherent
information, discussed previously in~\cite{Siddhu21}~(see
also~\cite{SinghDatta22}) can be easily extended to a channel's
reverse-coherent information as follows. If an $\ep$ $\log$-singularity in
$S_a(\ep)$ is stronger than the one in $S_c(\ep)$ then the reverse coherent
information $I_r(\ep)>0$.

\section{Damping-Dephasing Channel}

A qubit density operator can always be written in the Bloch form,
\begin{equation}
    \rho = \frac{1}{2} (\Ibb_2 + x X + y Y + zZ), 
    \label{eq:blochIn}
\end{equation}
where the Bloch vector $\rB = (x,y,z)$ has Euclidean norm at most one, $\Ibb_2$
is the $2 \times 2$ identity matrix, $X =
\dyad{0}{1} + \dyad{1}{0}$, $Y= i (\dyad{0}{1} - \dyad{1}{0})$, and $Z =
\dya{0} - \dya{1}$ are Pauli matrices. 
The von-Neumann entropy of $\rho$ in~\eqref{eq:blochIn} is $h \big( (1 +
|\rB|)/2 \big)$ where $h(p):= -\big( p \log p + (1-p) \log (1-p) \big)$ is the
binary entropy function.

\subsection{Dephasing channel}
\label{sec:Dephchannel}

An isometry $F: \HC_a \mapsto \HC_b \ot \HC_{c1}$ of the form,
\begin{equation}
    F \ket{0} = \ket{0} \ot \ket{\phi_0} \quad \text{and} \quad
    F \ket{1} = \ket{1} \ot \ket{\phi_1},
    \label{eq:dephDef}
\end{equation}
where $\ket{\phi_i} = \sqrt{1-p} \ket{+} + (-1)^i \sqrt{p} \ket{-},$ $i \in
\{0,1\}$, $0 \leq p \leq 1/2$, and $\ket{\pm} = (\ket{0} \pm \ket{1})\sqrt{2}$
generates channels~(via~eq.\eqref{eq:chanDef})
\begin{align}
    \begin{aligned}
    \DC_p(M) &= (1-p) M + p Z M Z, \\
    \DC_p^c(M) &= \phi_0 \Tr(M U_0) + \phi_1 \Tr(M U_1 ),
    \label{eq:dephasing}
    \end{aligned}
\end{align}
where $\phi_i = \dya{\phi_i}$ and $U_i = \dya{i}$.
Here $\DC_p$ is the qubit dephasing channel with dephasing probability $p$.
Let $\rho_a$ be a qubit density operator with Bloch vector $\rB_a = (x, y, z)$
then $\rho_b := \DC_p(\rho_a)$, and $\rho_{c1} := \DC_p^c(\rho_a)$ have Bloch
coordinates
\begin{align*}
    \rB_b &= ( (1-2p)x, (1-2p)y, z), \quad \text{and} \\
    \rB_{c1} &= (1-2p, 0, 2 \sqrt{p(1-p)}z),
\end{align*}
respectively.
Channel $\DC_p$ is degradable for $0 \leq p \leq 1/2$ and EB at $p=1/2$. The
one-way, backward, and two-way capacities of this channel equal its coherent
information evaluated at $\Ibb_2/2$~\cite{Rains97},
\begin{equation}
    \EC(\DC_p) = \EC_{\leftarrow}(\DC_p) = \EC_{\leftrightarrow}(\DC_p)  = I_c(\DC_p, \Ibb_2/2) = 1- h(p).
\end{equation}
\subsection{Amplitude damping channel}
\label{sec:ADchannel}
An isometry $G: \HC_b \mapsto \HC_d \ot \HC_{c2}$ of the form,
\begin{align}
    \begin{aligned}
    G \ket{0} &= \ket{0} \ot \ket{1}, \\
    G \ket{1} &= \sqrt{g} \ket{0} \ot \ket{0} + \sqrt{1-g} \ket{1} \ot \ket{1},
    \label{eq:dampDef}
\end{aligned}
    \end{align}
where $0 \leq g \leq 1$ generates a pair of channels 
\begin{align}
    \begin{aligned}
    \AC_g(\rho) &= A_0 \rho A_0^{\dag} + A_1 \rho A_1^{\dag}, \\
    \AC_g^c(A) &= B_0 \rho B_0^{\dag} + B_1 \rho B_1^{\dag},
    \label{eq:damping}
\end{aligned}
    \end{align}
where $\AC_g$ is a qubit amplitude damping channel which damps $\dya{1}$ to
$\dya{0}$ with probability $g$ and $\AC^c$ is a qubit amplitude damping channel
with damping probability $1-g$ if one interchanges $\ket{0}_{c2}$ and
$\ket{1}_{c2}$.
An input density operator with Bloch vector $\rB_b = (x, y, z)$ is mapped to
$\rho_d := \AC_g(\rho_b)$, and $\rho_{c2} := \AC^c_g(\rho_b)$ with Bloch
vectors
\begin{align}
    \begin{aligned}
        \rB_d &= ( \sqrt{1 - g}x, \sqrt{1-g}y, (1-g)z + g), \quad \text{and} \\
    \rB_{c2} &= (\sqrt{g}x, -\sqrt{g}y, g -gz -1),
    \end{aligned}
\end{align}
respectively.

At $g = 1$, $\AC_g(M) = \Tr(M) \dya{0}$ and thus EB;
in general, $\AC_g$ is degradable when $0 \leq g < 1/2$ and anti-degradable
otherwise~\cite{WolfPerezGarcia07}. 
Thus, the coherent information $I_c(\AC_g) = \EC(\AC_g)$. In general
$I_c(\AC_g)$ is less than the reverse coherent information $I_r(\AC_g)$, for $g
\neq 0$,
\begin{equation}
    \label{eq:StrictneqA}
    \EC(\AC_g) < \EC_{\leftarrow}(\AC_g).
\end{equation}
Proof for this numerically observed fact~\cite{GarciaPatronPirandolaEA09} 
is shown in Lemma~\ref{lm:maxI_r}.
\subsection{Joint damping-dephasing channel}
\label{sec:Jointchannel}
The combined action of $\DC_p$ and $\AC_g$, given by either $\DC_p \circ \AC_g$
or $\AC_g \circ \DC_p$ as the action of the channels commute, gives a channel
$\FC$.  This channel has two parameters, the dephasing~($p$) and damping~$(g)$
probabilities.
Channel $\FC$ takes the form,
\begin{equation}
    \FC(\rho) = \sum_i O_i \rho O_i^{\dag},
    \label{eq:dampDeph}
\end{equation}
where $O_0 = \sqrt{1-p}(\ketbra{0}+\sqrt{1-g}\ketbra{1})$, $O_1= \sqrt{g}
\ket{0}\bra{1}$ and $O_2 = \sqrt{p}(\ketbra{0}-\sqrt{1-g}\ketbra{1})$.
These Kraus operators either commute or anti-commute with $Z$, $O_j Z_a =
(-1)^j O_j$, $j \in \{0,1,2\}$. The qubit output of $\FC$, $\rho_d$, has Bloch
vector,
\begin{equation}
    \rB_d = \big( (1-2p)\sqrt{1-g}x,  (1-2p)\sqrt{1-g}y, (1-g)z + g \big).
    \label{eq:outBloch}
\end{equation}
Sometimes it is convenient to write the channel's input and output in the following form:
\begin{align}
    \begin{aligned}
        \rho_a &= 
    \begin{pmatrix}
        1 - \rho_{11} & \rho_{01} \\
        \rho_{01}^* & \rho_{11}
    \end{pmatrix}, 
    \\
    \rho_d = \FC(\rho_a) &= 
    \begin{pmatrix}
        1-\rho_{11}e^{-t/T_1} & \rho_{01}e^{-t/T_2} \\
        \rho_{01}^*e^{-t/T_2} & \rho_{11}e^{-t/T_1} 
    \end{pmatrix},
    \label{eq:simplRho}
\end{aligned}
    \end{align}
where $\rho_{01}$ is a complex parameter, and $ \rho_{11}, \ t, \ T_1, \ T_2$
are real non-negative parameters. These parameters are related to $g$ and $p$
as
\begin{equation}
    g = 1 - e^{-t/T_1} \quad 
    \text{and} \quad
    p = \frac{1}{2}(1 - e^{-t(1/T_2 - 1/2T_1)}),
\end{equation}
where the constraint $0 \leq p \leq 1/2$ implies that $T_2 \leq 2T_1$ with
equality when $p = 0$, i.e., dephasing noise is absent and $\FC$ is simply a
qubit amplitude damping channel. Here $T_1$ can be viewed as a decay
constant~(sometimes called single-qubit relaxation time), with which the 
$\ket{1}$ state decays to $\ket{0}$ and $T_2$
as the decay constant~(sometimes called the dephasing time), with which the
off-diagonal terms in $\rho_a$ dephase.  This relationship of $\FC$ with the
parametrization in~\eqref{eq:simplRho} and parameters $T_1$ and $T_2$ make $\FC$
amenable for describing noise in physical setups.

It is useful to write the complement of $\FC$ in two different ways.  The first
makes use of the relationship between the Kraus operators~\eqref{eq:dampDeph}
of $\FC$ and that of its complement (see above \eqref{eq:chanDef}) to give,
\begin{equation}
    \FC^c(A) = P_0 A P_0^{\dag} + P_1 A P_1^{\dag}
    \label{eq:chanComp1}
\end{equation}
where
\begin{align}
    \begin{aligned}
     P_0 &= \sqrt{g} \dyad{0}{1} + \sqrt{1-p} \dyad{1}{0} + \sqrt{p} \dyad{2}{0},
    \\
    P_1 &= \sqrt{(1-g)(1-p)} \dyad{1}{1} - \sqrt{p(1-g)} \dyad{2}{1},
\end{aligned}
\end{align}
and $\FC^c$ maps $a$ to an environment $e$.
Alternatively, notice the isometry, $H: \HC_a \mapsto \HC_d \ot \HC_{c1c2}$,
\begin{equation}
    H = (I_{c1} \ot G) F,
\end{equation}
defines the channel $\FC(A) = \Tr_{c1c2} (HAH^{\dag})$~\eqref{eq:dampDeph} and
its complement $\GC(A) = \Tr_{d} (HAH^{\dag})$. Using this second definition, we
obtain
\begin{equation}
    \GC(A) = Q_0 A Q_0^{\dag} + Q_1 A Q_1^{\dag},
    \label{eq:chanComp2}
\end{equation}
here Kraus operators $Q_i: \HC_{a} \mapsto \HC_{c1c2}$ take the form
\begin{align}
    \begin{aligned}
    Q_0 &= (\ket{\phi_0}_{c1} \ot \ket{0}_{c2}) \bra{0} + \sqrt{g}
    (\ket{\phi_1}_{c1} \ot \ket{1}_{c2}) \bra{1}
    \\
    Q_1 &= \sqrt{1-g} (\ket{\phi_1}_{c1} \ot \ket{0}_{c2}) \bra{1}
    \end{aligned}
\end{align}
where $\ket{\phi_i}$ are defined below~\eqref{eq:dephDef}.
Using the form~\eqref{eq:simplRho} for the input, the output 
$\rho_{c1c2} = \GC(\rho_a)$ can be written as a block matrix 
\begin{equation}
    \begin{pmatrix}
        (1-g) \rho_{11} \proj{\phi_1} + (1- \rho_{11})\proj{\phi_0}  &
        \sqrt{g} \rho_{01}\dyad{\phi_0}{\phi_1}  \\
        \sqrt{g} \rho_{10} \dyad{\phi_1}{\phi_0}  & g \rho_{11}\proj{\phi_{1}} 
    \end{pmatrix},
    \label{eq:blochMatC2}
\end{equation}
where each block is $2 \times 2$.

\section{Entanglement Distillation over the Joint Damping-Dephasing Channel}
\subsection{Forward only Distillation}

A lower bound on the one-way distillable entanglement of $\FC$ is the channel's
coherent information. We first give bounds on parameters for which the coherent
information is positive.

\lemma{The coherent information $I_c(\FC)$ is strictly positive for
\begin{equation}
    g < g_{\max}(p):= 1 - \frac{1}{2(1 - 2 p (1-p))}
\end{equation}
where $0 \leq p < 1/2$.  \label{lm:q1Pos}} 
\begin{proof}
    Consider an input density operator $\rho_a(\ep)$ of the form
    in~\eqref{eq:simplRho} where $\rho_{01} = 0$, $\rho_{11} = \ep$, and
    $0 \leq \ep \leq 1$.  Channel output $\rho_d(\ep)$ has eigenvalues
    $\{\ep (1-g), 1 - \ep(1-g) \}$, and thus $S_d(\ep)$ has an $\ep$ 
    $\log$-singularity of rate $x_d = 1-g$.
    The output $\rho_e(\ep) = \FC^{c}\big( \rho_a(\ep) \big)$
    in~\eqref{eq:chanComp1} has eigenvalues
    \begin{equation}
        \lm_0 = g \ep, \quad 
        \lm_{\pm} = \frac{1 - g \ep}{2}(1 \pm \sqrt{u^2 + v^2})
    \end{equation}
    where $u = 2 \dl \sqrt{p(1-p)}, v = 1-2p$ and $\dl = (1 + \ep
    (g-2))/(1-\ep g)$. While $\lm_{+}$ is non-zero at $\ep = 0$, expanding
    $\lm_{-}$ to linear order gives $\lm_{-} = \ep 4 p (1-p) (1-g) + O(\ep^2)$. 
    As a result $S_e(\ep)$ has an $\ep$ $\log$-singularity
    of rate $x_{e} = g + 4 p (1-p) (1-g)$.

    If $x_d > x_e$ then $I_c(\FC)>0$. The inequality $x_d > x_e$ occurs where
    $g < g_{\max}:= 1 - \frac{1}{2(1 - 2 p (1-p))}$

\end{proof}

Next, we show that the calculation of $I_c(\FC)$ can be simplified by the
following observation:
\lemma{The coherent information $I_c(\FC, \rho)$ at an input operator
with Bloch vector $\rB = (x, y, z)$ only depends on $z$ and $x^2 + y^2$.
\label{lm:xzOpt}}

Proof for the Lemma in App.~\ref{App.ChnCoh} makes use of an alternate form
of the channel's complement~(see Sec.~\ref{sec:Jointchannel}).  The
simplification in Lemma~\ref{lm:xzOpt} allows calculation of $I_c(\FC)$ over
Bloch coordinates $\rB = (x,0,z)$. 
We numerically find this optimum to lie along $(0,0,z)$. In
Fig.~\ref{fig:forwardDistill} we plot the coherent information $I_c(\FC)$ as a
function of the amplitude damping probability $g$ for fixed dephasing
probability $p$. As $g$ is increased, the coherent information $I_c(\FC)$
decreases, becoming zero at $g = g_{\max}(p)$ and remaining zero thereafter.
Recently, an optimization of the coherent information over $z$ was also
reported in~\cite{MeleSalekEA24} in the context of bounding from below the
bosonic loss-dephasing channel's quantum capacity.

Using a log-singularity based argument we also prove that the
coherent information of the complementary channel is positive for a wide range
of parameters.

\lemma{The coherent information $I_c(\FC^c)$ is strictly positive for
$0 < p \leq 1/2$ when $0 < g < 1$.}
\label{lm:q1PosComp}

\begin{proof}
    Consider an input density operator $\rho_a$ of the form
    in~\eqref{eq:simplRho} where $\rho_{01} = 0$, $\rho_{11} = 1-\ep$, and $0
    \leq \ep \leq 1$. The density operator $\FC(\rho_a) = \rho_d$ has
    eigenvalues $\{g + \ep (1-g), (1-g) - \ep(1-g) \}$, and thus $S(\rho_d)$
    has no $\ep$ $\log$-singularity.
    The output $\FC^{c}(\rho)$ in~\eqref{eq:chanComp1} has eigenvalues,
    \begin{equation}
        \lm_0 = g (1-\ep), \quad \text{and} \quad
        \lm_{\pm} = \frac{1 - g (1-\ep)}{2}(1 \pm \sqrt{u^2 + v^2})
    \end{equation}
    where $u = 2 k \sqrt{p(1-p)}, v = 1-2p$ and $k = \big( \ep (2-g) - (1-g)
    \big)/\big( 1 - g(1-\ep) \big)$.  While $\lm_{+}$ is non-zero at $\ep = 0$,
    expanding $\lm_{-}$ to linear order in $\ep$ with $0 < g < 1$ gives
    $\lm_{-} = \ep 4 p (1-p) + O(\ep^2)$.  As a result $S\big( \FC^c(\rho)
    \big)$ has an $\ep$ $\log$-singularity of rate $x_{e} = 4 p (1-p)$.
    This rate is strictly positive for $0 < p \leq 1/2$ and thus
    $I_c(\FC^c)>0$.
\end{proof}

From the definition of $\FC$ and lower bounds on $\EC(\FC)$ it follows that
$\FC$ is (1) anti-degradable when $g \geq 1/2$ or at $p = 1/2$ (2) EB at $g =
1$ or $p = 1/2$ (3) degradable at $g = 0$, and also at $p=0$ and $g \leq 1/2$,
and (4) never degradable for all $0 < p \leq 1/2$ and $g \neq 1$ since
$I_c(\FC^c) > 0 $~(see Lemma~\eqref{lm:q1PosComp}). It follows from (4) that
$\FC$ is degradable iff $p=0$ and $g \leq 1/2$~(when $\FC = \AC_g$), or $g = 0$~(
when $\FC = \DC_p$). An alternate proof for this result appears in recent
work~\cite{MeleSalekEA24}. That same work shows $\FC$ is anti-degradable iff
$(1-2p)^2/\big( 1 + (1-2p)^2 \big) \leq g$ using techniques
from~\cite{ChenJiEA14}. In agreement with that result we supply
an explicit anti-degrading map $\Tilde{\mathcal{F}}$ for $g \geq\frac{(1-2p)^2}{1+(1-2p)^2}$ such that $\mathcal{F} =\Tilde{\mathcal{F}}\circ\mathcal{F}^c $.
Note that for all $0\leq p \leq 1$, 
$\frac{(1-2p)^2}{1+(1-2p)^2}\leq 1/2$. Thus, $g\geq 1/2$ implies $g \geq\frac{(1-2p)^2}{1+(1-2p)^2}.$
Let $$\ket{\kappa_k} = \sqrt{1-p}\ket{1}+(-1)^k\sqrt{p}\ket{2},\ket{\tilde{\kappa}_k} = \sqrt{p}\ket{1}+(-1)^k\sqrt{1-p}\ket{2}, \textnormal{ for } k \in \{0,1\}.$$
Define the map $\mathcal{R}$ via the following Kraus operators: 
$$
R_0 = \ket{1}\bra{0}+\ket{0}\bra{\kappa_0},
R_1 = \ket{1}\bra{\tilde{\kappa}_1}.$$
We construct the antidegrading map as a composition of the map $\mathcal{R}$ with a damping channel followed by a dephasing channel $\Tilde{\mathcal{F}} = \mathcal{D}_\delta\circ\mathcal{A}_\gamma\circ\mathcal{R},$ where the damping parameter $\gamma$, and the dephasing parameter $\delta$ are given by$$
\gamma = \frac{g-(1-2p)^2(1-g)}{1-(1-2p)^2(1-g)},\delta = \frac{1}{2}-\frac{(1-2p)\sqrt{1-(1-2p)^2(1-g)}}{2\sqrt{g}}.$$
One can verify that for  $1\geq g \geq\frac{(1-2p)^2}{1+(1-2p)^2},$ the damping and dephasing parameters define valid channel parameters with $0\leq \gamma, \delta \leq 1.$

\begin{figure}[htbp]
	\centering
    \includegraphics{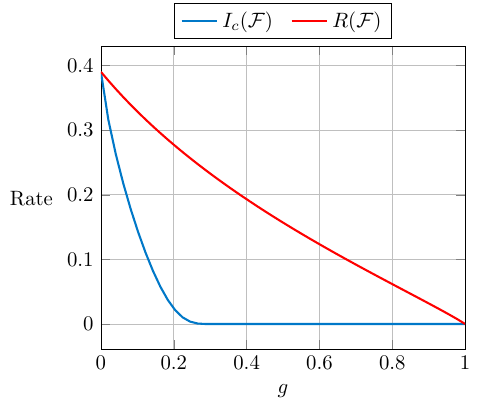}
    \caption{A plot of the coherent information $I_c$ of the damping-dephasing
    channel $\FC$~\eqref{eq:dampDeph} as a function of the amplitude damping
    probability $g$ for fixed dephasing probability $p=0.15$. The plot also
    shows the Rains information of the channel $R(\FC)$~\cite{FangFawzi21},
    which upper bounds the forward-assisted entanglement distillation rate.}
	\label{fig:forwardDistill}
\end{figure}

\subsection{Non-Additivity of Coherent information}

\begin{figure}[htbp]
	\centering
    \includegraphics[scale=.6]{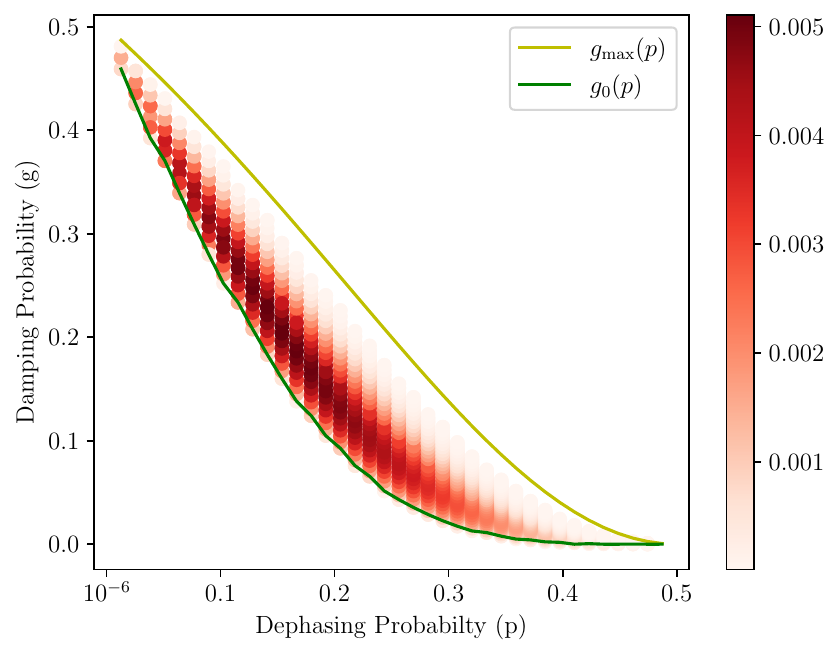}
    \caption{A phase diagram showing the region of dephasing and damping
    probabilities, $p$ and $g$, respectively, where the coherent information of
    $\FC$ is non-additive at the two-letter level. The amount of
    non-additivity found~\eqref{eq:nonAddAmt} is represented by the color.}
    \label{fig:nonAdd}
\end{figure}
We observe that non-additivity of the coherent information, i.e., strict inequality in
\begin{equation}
    \frac{1}{n} I_c(\FC^{\ot n}) \geq I_c(\FC),
    \label{eq:superAdd}
\end{equation}
where $n \geq 2$, occurs for a range of $p$ and $g$ values. When $n$ is the
smallest integer for which~\eqref{eq:superAdd} is a strict inequality, one says
non-additivity occurs at the $n$-letter level. We are able to observe
non-additivty at the lowest possible, two-letter level.
This observation comes from maximizing $I_c(\FC^{\ot 2}, \sg)$ over a simple
ansatz,
\begin{equation}
    \sg = \lm_1 \dya{00} + \lm_2\dya{\phi} + \lm_3 Z_{a1} \dya{\phi} Z_{a1} +
    \lm_4 \dya{11}
    \label{eq:ansatz}
\end{equation}
where $\phi = (\ket{00} + \ket{11})/\sqrt{2}$ is a maximally entangled state,
$Z_{a1}$ acts on the first input of the $\FC \ot \FC$ channel, and $\{\lm_i\}$
are free paramters that form a probability distribution.
The ansatz~\eqref{eq:ansatz} has a pleasing property,
at a special value of $\{ \lm_i\}$, $\lm_1 = (1+z)^2/4, \lm_2 = \lm_3 =
(1-z^2)/4,$ and $\lm_4 = (1-z)^2/4$ with $-1 \leq z \leq 1$, this ansatz
represents a product of two identical qubit density operators, each with Bloch
coordinates $\rB = (0,0,z)$. Since $I_c(\FC)$ is attained at these Bloch
coordinates~(see discussion below Lemma~\ref{lm:xzOpt}), the maximum of $I(\FC^{\ot 2},
\sg)$ over $\{\lm_i\}$, $I_c^*$, is at least $2 I_c(\FC)$.
Maximization over these parameters reveals that for a range of $0 < g <1$ and
$g_{0}(p) < p < g_{\max}(p)$ values, where $g_0(p)$ is found numerically,
non-additivity occurs at the two-letter level~(see Fig.~\ref{fig:nonAdd}). The
amount of non-additivity found is
\begin{equation}
    \dl = \frac{1}{2} I^*_c - I_c(\FC).
    \label{eq:nonAddAmt}
\end{equation}
For any fixed $p$ as $g$ is increased from zero, this amount is first zero until
$g$ reaches $g_0(p)$, then increases from zero, reaches a maximum, and
decreases to zero as $g$ approaches $g_{\max}(p)$~(see Fig.~\ref{fig:nonAdd2}).

\begin{figure}[htbp]
	\centering
    \includegraphics[scale=.5]{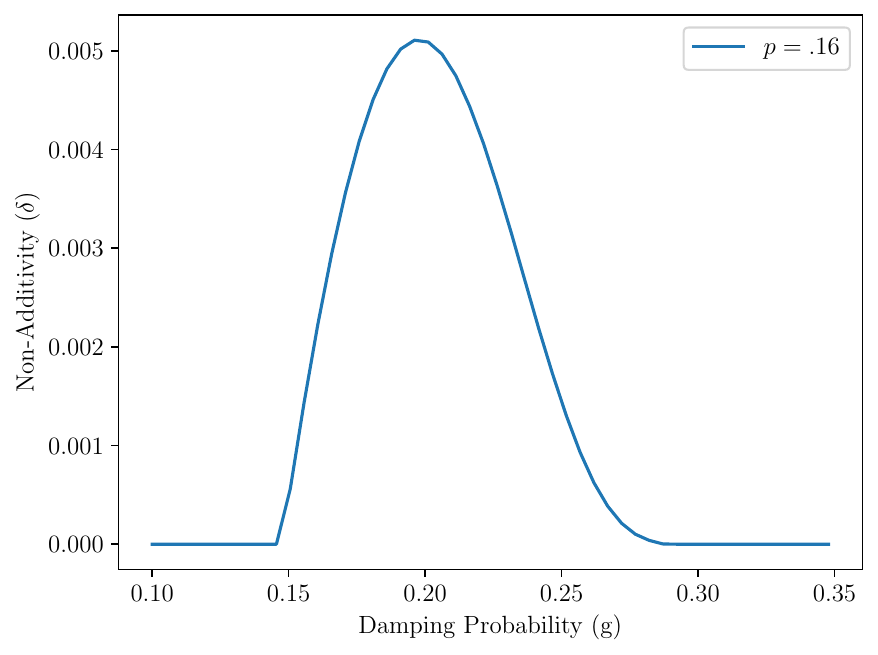}
    \caption{For fixed dephasing probability $p=.16$, the amount of
    non-additivity found, $\dl$~\eqref{eq:nonAddAmt} as a function of damping
    probability $g$.}
    \label{fig:nonAdd2}
\end{figure}

\subsection{Backward only distillation bound}

A lower bound on $\EC_{\leftarrow}(\FC)$ is the channel's reverse coherent
information.  We prove that it is positive for a
wide range of values $p$ and $g$.
\lemma{The reverse coherent information $I_r(\FC) > 0$ for all $0 \leq p <
1/2$ and $0 \leq g < 1$.  \label{lm:revCI}}

\begin{proof}
    Consider an input density operator $\rho_a(\ep) $ of the form
    in~\eqref{eq:simplRho} where $\rho_{01} = 0$, $\rho_{11} = \ep$ and $0< \ep
    \leq 1$.
    Thus $S\big( \rho_a(\ep) \big)$ has an $\ep$ $\log$-singularity of
    rate $x_a = 1$.

    As discussed in the proof of Lemma~\ref{lm:q1Pos}, $S(\FC^c(\rho))$ has an $\ep$
    $\log$-singularity of rate $x_{e} = g + 4 p (1-p) (1-g)$.  Notice $x_a >
    x_e$ whenever $p \neq 1/2$ and $g \neq 1$. As a result $x_a > x_e$.
    Since at $\ep = 0$ the revere coherent information is zero, for small $\ep$,
    $I_r(\FC)>0$.
\end{proof}

The above result is in agreement with~\cite{MeleSalekEA24} which arrives
at the positivity of the two-way capacity of $\EC$ using a different route.
To compute the reverse coherent information $I_r$, we use concavity of $I_r$ and a $Z$
symmetry in both $\FC$ and $\FC^c$.

\lemma{The reverse coherent information $I_r(\FC)$ is achieved at a qubit
density operator with Bloch coordinates $\rB = (0, 0, z)$.
\label{lm:xzOpt2}}

\begin{proof}
    Following an argument similar to one in the appendix of~
    \cite{garcia2008reverse}, we show the optimal input state for $I_r(\FC)$ is
    diagonal.
    Let $H' = \sum_i O_i \ot \ket{i}_e$ be an isometry that generates $\FC$
    in~\eqref{eq:dampDeph} and $\FC^c$ in~\eqref{eq:chanComp1}. This isometry
    satisfies $(W_e \ot Z_b) H'  = H' Z_a$ where $W = -\dya{0} + \dya{1} + \dya{2}$.
    This equality implies
    \begin{align}
        \begin{aligned}
            \FC(Z_a N Z_a) &= Z_b \FC(N) Z_b, \quad \text{and} \\
        \FC^c(Z_a N Z_a) &= W_e \FC^c(N) W_e
        \label{eq:symm}
        \end{aligned}
    \end{align}
    for any $N$.
    Let $\rho_a$ be any input, and $\rho_{be} = H' \rho_a (H')^{\dag}$.  The
    reverse coherent at this input, $ I_r(\FC, \rho_a) = S(\rho_a) - S\big(
    \FC^c(\rho_a) \big) =: S(b|e)_{\rho_{be}}$ is concave in $\rho_{be}$.
    For any $\Tilde{\rho}_a$, let 
    \begin{equation}
        \rho_a = \frac{\Tilde{\rho}_a + Z_a \Tilde{\rho}_a Z_a}{2},
    \end{equation}
    be a diagonal density operator.  From concavity and symmetry
    in~\eqref{eq:symm} it follows the reverse coherent information at this
    operator is larger than $I_r(\FC, \tilde{\rho}_a)$,
    \begin{align}
        \begin{aligned}
    I_r(\FC, \rho_a) = S(b|e)_{\rho_{be}} &\geq
        \frac{1}{2}[S(\Tilde{\rho}_a)-S\big(\FC^c(\Tilde{\rho}_a) \big) \\ 
            &+ S(Z\Tilde{\rho}_aZ)-S\big(\FC^c(Z\Tilde{\rho}_aZ)
        \big)]\\
            &=\frac{1}{2}[I_r(\FC, \Tilde{\rho}_a)+I_r(\FC, \Tilde{\rho}_a) ] \\ 
            &= I_r(\FC, \tilde{\rho}_a).
    \end{aligned}
    \end{align}

\end{proof}

In Fig.~\ref{fig:ciChan} the reverse coherent information $I_r(\FC)$ is
plotted as a function of the amplitude damping probability $g$ for fixed
dephasing probability $p$. As $g$ is increased, the reverse coherent
information $I_r(\FC)$ decreases, and for small $\dl g := 1-g$ it
takes an asymptotic form,
\begin{equation}
    I_r(\FC) \simeq \dl g (1-q) \big( q \dl g \big)^{q/(1-q)},
    \label{eq:asyRev}
\end{equation}
where $q := 4p(1-p)$.
This form comes from noticing numerically that for any fixed $p$,
as $g$ tends to one and $\dl g \mapsto 0$, the optimum reverse coherent
information is obtained at a density operator $\rho = \text{diag}(1 - \ep,
\ep)$ where $\ep$ is small. For such small $\ep$, the reverse coherent
information can be written as
\begin{equation}
    I_r(\ep) \simeq f(\ep) = (\al \ep \ln \ep + \bt \ep) \log_2 e 
    \label{eq:icSmEp}
\end{equation}
where terms of order $O(\ep^2)$ are dropped,
\begin{equation}
    \al = a_0(q) \dl g,  \quad 
    \bt = b_1(q) \dl g + b_2(q) \dl g \ln \dl g,
\end{equation}
where terms of $\ln(1 - \dl g)$ are dropped and 
\begin{equation}
    a_0(p) = q - 1, \quad
    b_1(p) = 1 - q + q \ln q, \quad \text{and} \quad
    b_2(p) = q,
    \label{eq:consts}
\end{equation}
where $q := 4p(1-p)$. Since $0 \leq p \leq 1/2$, $0 \leq q \leq 1$ and 
$\al < 0$, thus $I_r(\ep)>0$ for small enough $\ep$. The function $f(\ep)$
has a maximum value at $\ep = \ep^* := \exp\big( - (1 + \bt/\al)\big)$.  
This maximum value of~\eqref{eq:icSmEp} for small $\dl g $ can then be written as
\begin{equation}
    I_r(\FC) \simeq \dl g (1-q) \big( q \dl g \big)^{q/(1-q)}.
    \label{app:eq:asyRev}
\end{equation}
We find good numerical agreement between the left hand side and the right
hand side of the above equation for small $\dl g$~(recall $q = 4p(1-p)$),
where the left hand side is evaluated using direct numerics.

\lemma{For the amplitude damping channel $\AC_g$, the reverse coherent
information $I_r(\AC_g) \geq I_c(\AC_g)$ for all $0 \leq g \leq 1$.
\label{lm:maxI_r} }
\begin{proof}
    Notice at $g = 0$, equality holds as $I_r(\AC_g) = I_c(\AC_g) = 0$.  Let
    $I_r(z)$ and $I_c(z)$ denote $I_r(\AC_g, \rho)$ and $I_r(\AC_g, \rho)$,
    respectively, where $\rho$ has Bloch vector $\rB=(0,0,z)$.
    For $1/2 \leq g < 1$ one can show, using expression
    below~\eqref{eq:blochIn}, that $I_r(0) > 0$. On the other hand when $1/2
    \leq g < 1$, $I_c(\AC_g) = 0$, from discussion above
    eq.~\eqref{eq:StrictneqA}.
    For $0 \leq g <1/2$, $I_c(z)$ is maximum for $z \geq 0$ since
    \begin{align*}
        I_c(z) &= S(\AC_g(\rho_z))-S(\AC^c_g(\rho_z))\\
        &=h\left(\frac{1+|(1-g)z+g|}{2}\right)-h\left(\frac{1+|g-gz-1|}{2}\right)\\
        & = h\left((1-g)\frac{1-z}{2}\right)-h\left(g\frac{1-z}{2}\right)
\end{align*}
    and the maximum of the difference in the last line above is obtained
    when $z \geq 0$~(see Lemma~\ref{lm:entropyineq2}).
    For $z \geq 0$,  it is easy to check, using expression
    below~\eqref{eq:blochIn}, that $I_r(z) - I_c(z) \geq 0$

\end{proof}

\begin{lemma} \label{lm:entropyineq2}
For $0\leq\alpha\leq 1/2$, let the maximum of the expression $\max_{0\leq
    p\leq1} h((1-\alpha)p)-h(\alpha p),$ be attained at $\tilde{p}$, then
    $\tilde{p}\leq \frac{1}{2}. $
\end{lemma} 

\begin{proof}
    Let $f(\alpha,p) =h((1-\alpha)p)-h(\alpha p)$, we can show that for
    $0\leq\alpha\leq 1/2$, $p\geq 1/2$, $f$ is monotonically decreasing in $p$. This
    can be shown by computing $\frac{\partial f}{\partial p}$ and showing that
    it is negative for $p\geq 1/2.$ Then, the maximum of the function must be
    attained for $p\in[0,1/2].$
\end{proof}

\subsection{Improved backward only distillation protocol}
We introduce a two-stage protocol for distillation which we argue exceeds the
bound established by the reverse coherent information.  The first stage is a
modification of a standard recurrence protocol and the second stage does
hashing~\cite{BennettBrassardEA96}.
Let $\HC_{ri}$ and $\HC_{ai}$, $1 \leq i \leq 2$ be qubit Hilbert spaces, such
that $\HC_r := \HC_{r1} \ot \HC_{r2}$ is a reference space to inputs $\HC_{a}
:= \HC_{a1} \ot \HC_{a2}$ of $\FC \ot \FC$, which each map $\FC:
\HC_{ai}\rightarrow \HC_{di}$.  Let $\MC_d$ represent the operation applying
the controlled not unitary $U_{d} = \dya{0}_{d1} \ot \Ibb_{d2} + \dya{1}_{d2}
\ot X_{d2}$ followed by a standard basis $Z$~measurement on $\HC_{d2}$
resulting in outcome $(-1)^{i_{d2}}$, which for simplicity we refer to as
outcome~$i_{d2}$.

Initially, Alice prepares a state
\begin{equation}
    \ket{\psi_1}=\frac{1}{\sqrt{2}}\left(\ket{0}\ket{01}+\ket{1}\ket{10}\right)_{r1,a1a2}.
\end{equation}
This is equivalent to a single Bell pair between the $r1$ and $a1a2$ systems,
where the half of the Bell pair on the $a1a2$ system is encoded in a small
2-qubit code. Alice then sends the qubits $a1,a2$ over two copies of the
damping-dephasing channel $\mathcal{F}\otimes\mathcal{F}.$ The resulting state
is $\psi_{2}=\sum_{i,j} O_i \otimes O_j \ket{\psi_1}\bra{\psi_1}O_i^\dagger
\otimes O_j^\dagger \in \HC_{r1}\otimes \HC_{d1} \otimes \HC_{d2},$ as given in
Eq.~\eqref{eq:dampDeph}. Bob then applies a controlled-not gate from $d1$ to
$d2$ and measures $d2$ in the $Z$ basis, obtaining the outcome
$(-1)^{i_{d2}}Z$, accepting the state only if $i_{d2}=1$. Bob sends the outcome
$i_{d2}$ to Alice using backward communication, and they have agreement to only
keep pairs where Bob has measured~$i_{d2}=1$. 

We must then calculate the probability of the given outcome, as well as the
final state. To simplify the analysis of the consequence of this measurement on
the state~$\psi_2$, we note that performing a controlled-not followed by a $Z$
measurement on the second qubit, is equivalent to measuring the
observable~$Z_{d1} \otimes Z_{d2}$ prior to the action of the controlled-not
gate. This is due to the fact that under the action of the controlled-not gate
$U_{d} = \dya{0}_{d1} \ot \Ibb_{d2} + \dya{1}_{d2} \ot X_{d2}$, the following
holds: $U_d (Z_{d1} \otimes Z_{d2})U_d^{\dagger} = I_{d1} \otimes Z_{d2}$ and
therefore any measurement of $Z_{d_2}$ after the action of $U_d$ is equivalent
to the measurement of $Z_{d1} \otimes Z_{d2}$ prior to $U_d$.  The
post-measurement state given that $i_{d2}=1,$ is
$$\psi_3=\frac{\Pi_{d1,d2}{\psi_2}\Pi_{d1,d2}}{\Tr{[\Pi_{d1,d2}{\psi_2}]}},\textnormal{
    where }\Pi_{d1,d2} = \frac{\mathbbm{1}- Z_{d1}Z_{d2}}{2}.$$ Bob obtains the
outcome $i_{d2}=1$ with probability $\Tr{[\Pi_{d1,d2}{\psi_2}]}$, for which we
give a closed form below.

Note that $\Pi_{r1,r2} O_i \otimes O_j \ket{\psi_1} = c_{ij} O_i \otimes O_j
\ket{\psi_1}$, where $c_{ij} = 1$ if $O_i \otimes O_j$ commutes with
$Z_{d1}Z_{d2}$ and $0$ otherwise. Thus, the only terms that contribute to
$\psi_3$ are those resulting from applying $O_0 \otimes O_0, O_1\otimes
O_1,O_2\otimes  O_2,O_0\otimes  O_2,O_2\otimes  O_0$ to $\ket{\psi_1}.$ Then,
we find that the state is:
\begin{align*}
    \left((1-p)^2+p^2\right) \ket{\psi_1}\bra{\psi_1}+2p(1-p)Z_{d1}\ket{\psi_1}\bra{\psi_1}Z_{d1},
\end{align*} 
where we have replaced the $a1,a2$ labels with $d1,d2$ given the labelling of
the output space. Moreover, we find the probability of obtaining the above
state to be: $\Tr{[\Pi_{d1,d2}{\psi_2}]} = (1-g)$, again by keeping the $O_i
\otimes O_j$ terms that commute with $Z_{d1}Z_{d2}$.

To return back to the correct basis, we apply a controlled-not $d1$ to $d2$,
and get the state $\rho_{r1d1}\otimes \ket{1}\bra{1}_{d2}$, with
\begin{align}
    \rho_{r1d1}= (1-q/2) \phi_{r1 d1}+\frac{1}{2}q Z_{d1}\phi_{r1 d1}Z_{d1},
    \label{eq:psState}
\end{align}
where $q = 4p(1-p)$ and $\phi = \ketbra{\phi}$ is the maximally entangled
state. The overall success probability of the modified recurrence step of the
protocol is $p_s = (1-g)$. This step has practical relevance, the resulting
state~\eqref{eq:psState} is just a dephased version of the input without any
amplitude damping even though the input was sent via $\FC$ that applies both
dephasing and damping. 
It should be noted that in this stage, replacing the maximally entangled state
$\ket{\psi_1}$ with some $\ket{\psi_s} = \sqrt{s}\ket{0}\ket{01} + \sqrt{1-s}
\ket{1}\ket{10}$, $0 \leq s \leq 1$ leaves the output state~\eqref{eq:psState}
unchanged while modifying $p_s$ to a possibly lower value of $4s(1-s)(1-g)$.

We also note that this state uses two channel uses in order to produce a single
Bell pair (upon successful measurement), thus it has an ideal rate of $R =
1/2.$ Following this step, the hashing protocol is carried out on several
copies of the accepted state $\rho_{r1d1}$, resulting in a protocol with an
overall yield of: 
\begin{equation}
    Y = \frac{1}{2}p_s(1-h(q/2)).
    \label{eq:yldNow}
\end{equation}
Hashing requires one-way communication which can always be done from receiver
to sender. Thus, all classical communication in the protocol occurs from
receiver to sender.
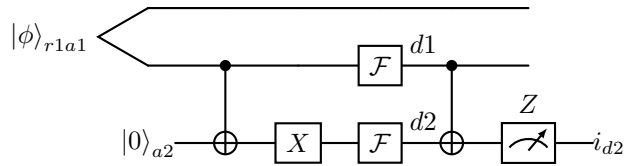
\begin{figure}[h]
\centering
\begin{quantikz}
\makeebit[angle=-60]{$\ket{\phi}_{r1a1}$}& & & & &
\\
&\ctrl{1} & & \gate[label style={label={10:$d1$}}]{\mathcal{F}}& \ctrl{1}& \\ 
\ket{0}_{a2} &  \targ{} & \gate{X} & \gate[label style={label={10:$d2$}}]{\mathcal{F}} & \targ{} & \meter{Z} &{i_{d2}}
\end{quantikz}
\caption{The modified recurrence stage of the proposed distillation protocol.}
\end{figure}

The distillation protocol can also be viewed as preparing a GHZ state, 
$(\ket{000} + \ket{111})/\sqrt{2}$, across the $r \ot
a1 \ot a2$ system, flipping one of the qubits on the $a1 \ot a2 $ system, passing the
$a1 \ot a2$ system across the damping-dephasing channel and checking the parity between
the two systems at the channel output. This parity is only flipped if one or 
both the noisy qubits experience damping and thus damping can be detected exactly.
This parity check is done by measuring one of the output qubits, leaving a Bell
state between the reference $r$ and the channel output. In this way, the protocol
uses a GHZ state to protect one Bell pair worth of entanglement from damping error.

We find~(see Fig.~\ref{fig:ciChan}) for both modest and specially high noise
regimes, this yield is higher than the channel coherent information and the
channel reverse coherent information. 
\begin{figure}[htbp]
	\centering
    \includegraphics{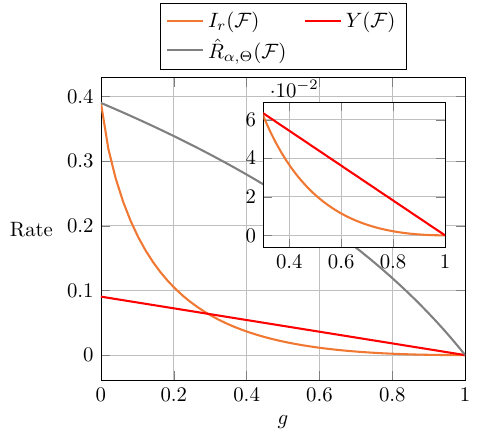}
    \caption{A plot of the reverse coherent information $I_r$ and yield
    $Y$~\eqref{eq:yldNow} of the damping-dephasing channel
    $\FC$~\eqref{eq:dampDeph} as a function of the amplitude damping
    probability $g$ for fixed dephasing probability $p=0.15$. The plot also
    shows the 
    %
    %
    geometric R\'{e}nyi Theta-information of the channel
    $\hat{R}_{\alpha,\Theta}(\FC)$~\cite{FangFawzi21}~$(\al = 1 +
    2^{-20})$, which upper bounds the two-way assisted entanglement
    distillation rate. The inset plot zooms into the $g \in [0.3,1]$ interval
    where $Y$ exceeds $I_r$.}
	\label{fig:ciChan}
\end{figure}
For any $p > 0$, the yield $Y$ exceeds the asymptotic
estimate~\eqref{eq:asyRev} of $I_r(\FC)$ for small enough $\dl g$, i.e.,
the ratio
\begin{align}
    \begin{aligned}
        \frac{Y}{I_r(\FC)} &\simeq \frac{1}{(\dl g)^{q/(1-q)}}
    \left( \frac{1 - h(q/2)}{2(1-q)q^{q/(1-q)}} \right),
    \label{eq:ratio}
    \end{aligned}
\end{align}
can always be made larger than one for small enough $\dl g$.  As a result, $Y$
can provide a strictly tighter lower bound on the backward capacity of the
damping-dephasing channel beyond the state-of-the art rates given by the
reverse coherent information.
\section{Evaluation of the Results}
To benchmark our proposed scheme's yield, we would like to find the tightest
possible upper bound on the two-way assisted capacity of the dephasing-damping
channel. 
Figure \ref{fig:ub_comparison} shows a comparison of various upper bounds used
in the literature.  One upper bound shown in the plots is one half of the
quantum mutual information of the channel which is an upper bound on the
two-way assisted quantum capacity~\cite{KhatriWilde20}.  The quantum mutual
information of the channel is given by the following optimization 
$$\max_{\phi_{AA^\prime}}I(A:B)_{\sigma}, \quad \text{where} \quad \sigma =
\mathcal{N}_{A^\prime\rightarrow B}(\phi_{AA^\prime}).$$
The quantum mutual information of the channel at $\rho_A$~\cite{BennettShorEA02}: 
$$I(\mathcal{N},\rho_A) = S(\rho_A)+S(\mathcal{N}(\rho_A))+\text{Tr}[W \log_2 W],$$
where $W_{ij} = \text{Tr}[E_i \rho_A E_j^\dagger]$, and $E_i,E_j$ are Kraus
operators of the channel. The term $-\text{Tr}[W \log_2 W]$ is the entropy
exchange. The quantum mutual information of the channel is concave in the
input density matrix $\rho_A$~\cite{BennettShorEA02}, and when the channel is
$Z$-covariant, we can take $\rho_A$ to be diagonal. Since the dephasing-damping
channel is $Z$-covariant, as shown in equation (\ref{eq:symm}), the evaluation
of the upper bound of one-half of the quantum mutual information of the channel
reduces to a simple optimization problem easily carried out using
\texttt{scipy.optimize.minimize}~\cite{2020SciPy-NMeth}.  The plot in Figure
\ref{fig:ub_comparison} also shows the \emph{max-Rains} information and the
\emph{geometric R\'{e}nyi Theta-information}, $\hat{R}_{\alpha,\Theta}$, upper
bounds on the two-way assisted capacity, each having a semidefinite program
formulation along with accompanying code provided by~\cite{FangFawzi21}.

From the plots in Fig.~\ref{fig:ub_comparison}, it is evident that the 
geometric R\'{e}nyi Theta-information, $\hat{R}_{\alpha,\Theta}$, 
provides the tightest upper bound on the two-way
assisted quantum capacity of the dephasing-damping channel. 

\begin{figure}[htbp]
	\centering
    \includegraphics{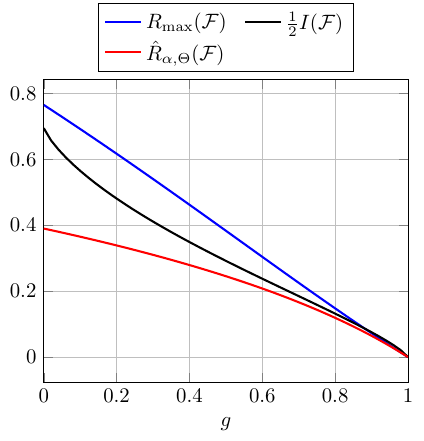}
    \caption{A comparison of various upper bounds from the literature for the
    two-way classical communication assisted capacity of the damping-dephasing
    channel at dephasing probability $p=0.15$.}
	\label{fig:ub_comparison}
\end{figure}

A combined lower bound on the achievable backward-communication assisted
entanglement distillation rate for different noise regimes is given by $L(\FC)
= \max\{I_r(\FC),Y\}$. At $g=0$, the channel is a pure dephasing channel, where
the lower bound $I_r(\FC)=I_c(\FC)=1-h(p)$ matches the upper bound given by the
geometric R\'{e}nyi Theta-information,
$\hat{R}_{\alpha,\Theta}(\FC)$~\cite{FangFawzi21}. We plot the lower bound
$L(\FC)$ and the upper bound $\hat{R}_{\alpha,\Theta}(\FC)$ against the
dephasing probability $p$, for different values of the amplitude damping
probability $g$, observing how the gap between the lower and upper bound
changes in Figure \ref{fig:gap_evol}, as well as the difference between the
upper and lower bounds in Figure \ref{fig:gap_diff_evol}.

 \begin{figure}[htbp]
           \centering
     \begin{subfigure}[b]{0.32\textwidth}
          \centering
          \resizebox{\linewidth}{!}{\includegraphics{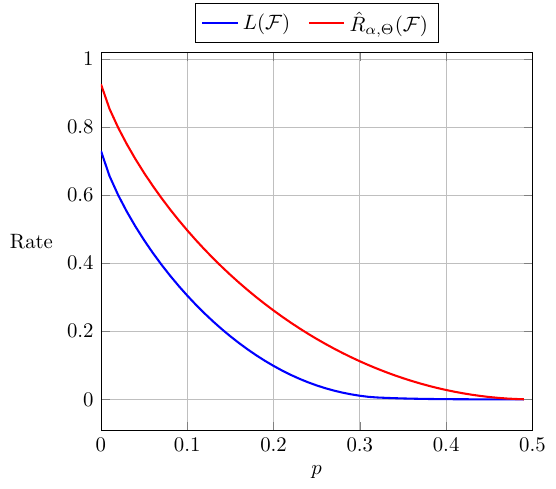}}  
          \caption{$g=0.1$}
          \label{fig:A}
     \end{subfigure}
     \begin{subfigure}[b]{0.32\textwidth}
          \centering
          \resizebox{\linewidth}{!}{\includegraphics{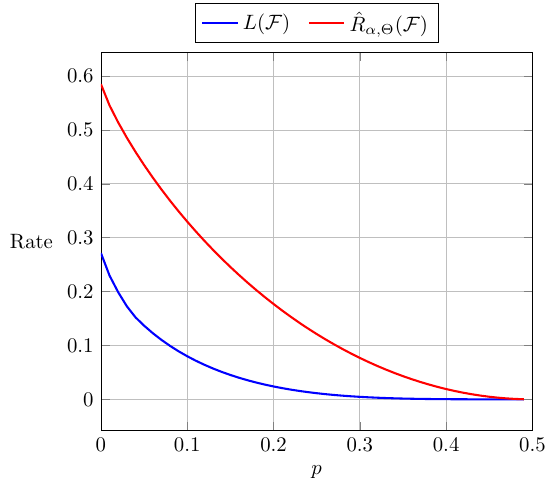}}  
          \caption{$g=0.5$}
          \label{fig:B}
     \end{subfigure}
     \begin{subfigure}[b]{0.32\textwidth}
          \centering
          \resizebox{\linewidth}{!}{\includegraphics{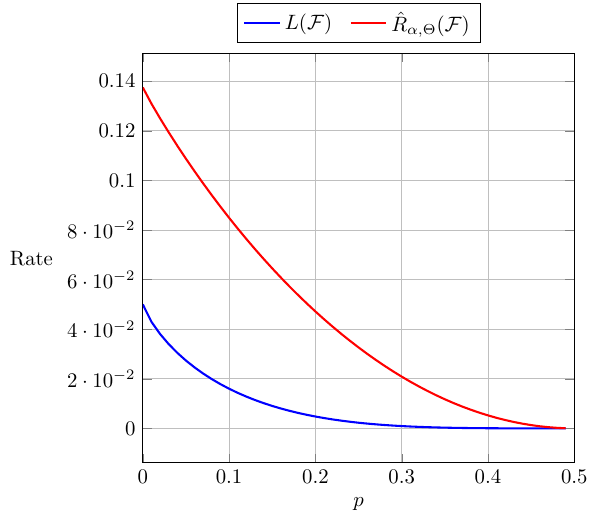}}  
          \caption{$g=0.9$}
          \label{fig:C}
     \end{subfigure}
     \caption{For different $g$, the plots show the 
     geometric R\'{e}nyi Theta-information upper bound
     $\hat{R}_{\alpha,\Theta}(\FC)$~\cite{FangFawzi21} with $\al = 1 + 2^{-20}$
     and the lower bound $L(\FC)$ for the joint damping-dephasing channel.}
     \label{fig:gap_evol}
 \end{figure}

  \begin{figure}[htbp]
           \centering

     \begin{subfigure}[b]{0.32\textwidth}
          \centering
          \resizebox{\linewidth}{!}{\includegraphics{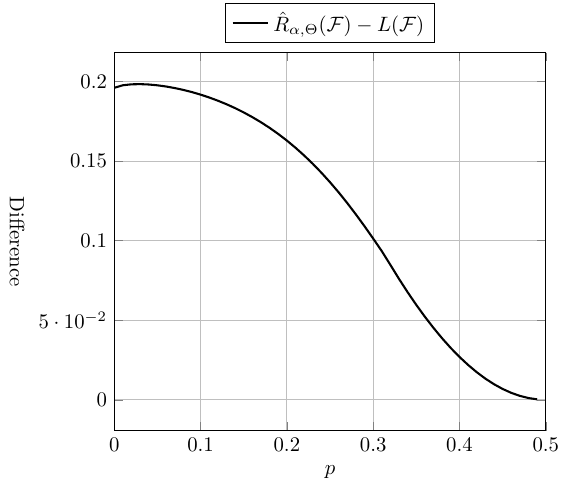}}  
          \caption{ $g=0.1$}
          \label{fig:8A}
     \end{subfigure}
     \begin{subfigure}[b]{0.32\textwidth}
          \centering
          \resizebox{\linewidth}{!}{\includegraphics{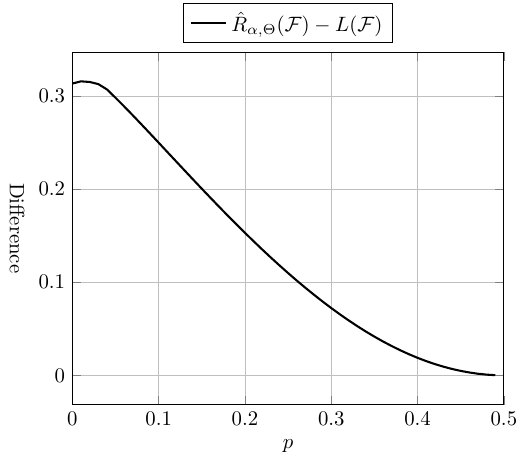}}  
          \caption{$g=0.5$}
          \label{fig:8B}
     \end{subfigure}
     \begin{subfigure}[b]{0.32\textwidth}
          \centering
          \resizebox{\linewidth}{!}{\includegraphics{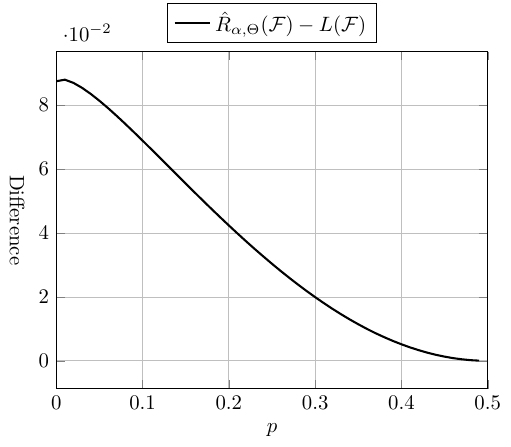}}  
          \caption{$g=0.9$}
          \label{fig:8C}
     \end{subfigure}
     \caption{For different $g$, the plots show the difference between the
     geometric R\'{e}nyi Theta-information upper bound
      $\hat{R}_{\alpha,\Theta}(\FC)$~\cite{FangFawzi21} with $\al = 1 +
      2^{-20}$} and the lower bound $L(\FC)$ for the joint damping-dephasing
      channel.
      \label{fig:gap_diff_evol}
 \end{figure}

\section{Discussion}
The feature of the proposed distillation protocol is to isolate the sources of
noise in order to address each individually. In the first stage, the modified
recurrence step identifies amplitude damping noise. In fact, by post-selecting
out the measurement of $i_{d2} = 1$, the scheme guarantees that any noise that
is a result of amplitude damping noise is caught. This exploits the asymmetric
source of the noise, that $\ket{1}$ can be mapped to~$\ket{0}$ but the inverse
is not true under this source of noise. Moreover, the dephasing noise commutes
with this first step, and can be mapped into a source of dephasing noise on the
post-selected state. The rate of the dephasing noise on this post-selected
qubit is: $2p(1-p)$, which is just the result of one of the receiver's qubits
having undergone a~dephasing $Z$ error. It is worth remarking that if both of
the receiver's qubits were dephased then those errors would cancel out. 

Our backward-only protocol has rates higher than those given by the
damping-dephasing channel $\FC$'s (reverse $I_r$)~coherent information~($I_c$)
when noise in this channel is modest. It would be interesting explore other
protocols, perhaps forward only, and appropriately compare them with $I_c$ and
$I_r$ in the low-noise regime of $\FC$.

In addition to these rates for entanglement sharing, it would be valuable to
study other quantum capacities of the physically well motivated
channel~$\FC$~( for instance see~\cite{LeditzkyKaurEA18}).  Our study of the
channel's quantum capacity reveals non-additivity in the channel's coherent
information.  This non-additivity can be found at the simplest two-letter level
using a neat and explicit ansatz~\eqref{eq:ansatz}.  The magnitude of the
non-additivity observed, $O(10^{-3})$, is comparable to those found for
dephrasure and generalized erasure channels~\cite{LeditzkyLeungEA18,
SiddhuGriffiths21, Filippov21} however the noise channel in our case is
qualitatively different. Our channel is strongly motivated to capture $T_1$ and
$T_2$ noise, both observed together in practice.  Occurrence of two-letter
level non-additivity in such a practically relevant setting paves the way for
its experimental study.

\section*{Acknowledgement}
T.J, V.S., and J.S, are supported by the U.S. Department of Energy, Office of
Science, National Quantum Information Science Research Centers, Co-design
Center for Quantum Advantage (C2QA) contract (DE-SC0012704). 
This work was done in part while the author was visiting the Simons Institute
for the Theory of Computing, supported by NSF QLCI Grant No.~2016245.
We also thank Francesco Anna Mele for bringing their recent work
~\cite{MeleSalekEA24} to our attention. 

\printbibliography[heading=bibintoc]

\appendix

\section{Channel Coherent information}
\label{App.ChnCoh}

We claimed in Lemma~\ref{lm:xzOpt} that the coherent information of $\FC$ at an input
density operator $\rho$ with Bloch coordinates $\rB = (x,y,z)$, $I_c(\FC,
\rho)$, only depends on $z$ and $x^2 + y^2$. Proof for this claim is as follows:

\begin{proof}
    The coherent information, $I_c(\FC, \rho) = S\big( \FC(\rho) \big) - 
    S\big( \GC(\rho) \big)$, where we use $\GC$ in~\eqref{eq:chanComp2} to
    be the complement of $\FC$.
    We show each term in this difference depends on $z$ and $x^2 + y^2$.
    The first term, $S \big( \FC(\rho)
    \big)$ is the entropy of a qubit density operator and this entropy,
    see comment below~\eqref{eq:entroVon}, 
    only depends on $|\rB_d|$ defined in~\eqref{eq:outBloch},
    Notice $|\rB_d|$ depends on $z$ and $x^2 + y^2$. We now analyze the second
    term, $S\big(\GC(\rho)\big)$. This term is the entropy of the Block matrix
    $\rho_{c1c2}$ in~\eqref{eq:blochMatC2}.  This entropy depends on
    eigenvalues of $\rho_{c1c2}$. These eigenvalues $\{\lm_i\}$ are solutions
    to the polynomial $g(\lm) = 0$ where
    \begin{equation}
        g(\lm) = \det(R), \;
        R = \rho_{c1c2} - \lm \Ibb_4 
        = \begin{pmatrix} 
            R_{11} & R_{12} \\
            R_{21} & R_{22}
        \end{pmatrix}.
    \end{equation}
    and the $2 \times 2$ blocks
    \begin{align}
        \begin{aligned}
            R_{11} &= \frac{1}{2} \big( (1-g) \dya{\phi_1} (1+z) + \dya{\phi_0} (1-
            z)\big) - \lm \Ibb_2, \\
            R_{12} &= \sqrt{g} \dyad{\phi_0}{\phi_1} (x + i y)/2 =
            R_{21}^{\dag}, \\
            R_{22} &= g \dya{\phi_{11}} (1-z)/2 - \lm \Ibb_2 
        \end{aligned}
    \end{align}
    Using Schur's formula for determinant for block matrices,
    \begin{equation}
        \det(R) = \det(R_{11}) \det(R_{22} - R_{12}^{\dag} R_{11}^{-1} R_{12} )
        \label{eq:detx}
    \end{equation}
    Notice $R_{ij}$ are all $2 \times 2$ matrices and for two such matrices $A$
    and $B$, the determinant $\det(A+B) = \det(A) + \det(B) + \Tr(A)\Tr(B) -
    \Tr(AB)$. Using this equality in the previous equation~\eqref{eq:detx},
    we get $ \det(R_{22} - R_{12}^{\dag} R_{11}^{-1} R_{12} )$ equals
    \begin{align}
        \begin{aligned}
            & \det(R_{22}) + \det(R_{12}^{\dag} R_{11}^{-1} R_{12}) \\
        &-\Tr(R_{22})\Tr(R_{12}^{\dag} R_{11}^{-1} R_{12}) +
        \Tr(R_{22} R_{12}^{\dag} R_{11}^{-1} R_{12})
        \end{aligned}
    \end{align}
    Notice each term on the right side of the equality is either independent of
    $x,y$ or depends on $x^2 + y^2$. Using this and the fact that $R_{11}$ only
    depends on $z$, we find $\det(R)$ in eq.~\eqref{eq:detx} depends on $x^2 +
    y^2$ and $z$.  This dependence implies $g(\lm)$ and thus its roots
    $\{\lm_i\}$~(the eigenvalues of $\rho_{c1c2}$) depend on $x^2 + y^2$ and
    $z$. 
\end{proof}

\end{document}